\documentclass[11pt,a4paper]{article}

\usepackage{fullpage}
\usepackage[utf8]{inputenc}
\usepackage{enumerate}
\usepackage{amssymb}
\usepackage{amsmath}
\usepackage{nicefrac}
\usepackage{graphicx}
\usepackage{bbm}
\usepackage{ae}
\usepackage[breaklinks,bookmarks=false]{hyperref}
\usepackage[usenames,dvipsnames,svgnames,table]{xcolor}
\usepackage{bbm}
\usepackage[amsmath,amsthm,hyperref,thmmarks]{ntheorem}
\usepackage[capitalise,noabbrev]{cleveref}
\usepackage{algorithm,algpseudocode}
\usepackage{tikz}
\usetikzlibrary{patterns}
\usetikzlibrary{decorations.pathreplacing}
\usepackage{mathtools}
\usepackage{dsfont}

\usepackage[backend=biber, style=alphabetic, backref=true,maxbibnames=99, url=false, doi=false,isbn=false]{biblatex}
\bibliography{refs} 
\hypersetup{
    colorlinks=true,
    linkcolor=violet,
    filecolor=magenta,      
    urlcolor=cyan,
    citecolor=blue,
    pdffitwindow=true,
}

\newtheorem{theorem}{Theorem}[section]

\newtheorem{lemma}[theorem]{Lemma}

\newtheorem{definition}[theorem]{Definition}
\newtheorem{fact}[theorem]{Fact}

\newtheorem{conjecture}[theorem]{Conjecture}

\qedsymbol{\openbox}

\usepackage{thmtools}
\usepackage{thm-restate}

\DeclareMathOperator{\spn}{span}
\DeclareMathOperator{\rk}{rank}

\DeclareMathOperator{\opt}{opt}

\newcommand{\bw}{\mathbf{w}}

\def\R{\mathbb{R}}
\def\bone{{\bf 1}}
\newcommand{\E}[1]{{\mathbb{E}}\left[#1\right]}

\newcommand{\bF}{\mathbb{F}}
\newcommand{\Ind}{\mathcal{I}}
\newcommand{\Alg}{\mathcal{A}}
\newcommand{\PM}{\mathcal{P}}
\def\BM{{\cal B}}
\def\cM{{\cal M}}
\hypersetup{ colorlinks=true, linkcolor=blue, filecolor=magenta, urlcolor=blue, }

\definecolor{arylideyellow}{rgb}{0.91, 0.84, 0.42}
\usepackage{authblk}

\title{Matroid Partition Property \\and the Secretary Problem}
\author{Dorna Abdolazimi\thanks{\href{mailto:dornaa@cs.washington.edu}{dornaa@cs.washington.edu}. Research supported by NSF grant CCF-1907845 and Air Force Office of Scientific Research grant FA9550-20-1-0212.}}
\author{Anna R. Karlin\thanks{\href{mailto:karlin@cs.washington.edu}{karlin@cs.washington.edu}. Research supported by Air Force Office of Scientific Research grant FA9550-20-1-0212 and NSF grant CCF-1813135.}}
\author{Nathan Klein\thanks{\href{mailto:nwklein@cs.washington.edu}{nwklein@cs.washington.edu}. Research supported in part by NSF grants DGE-1762114, CCF-1813135, and CCF-1552097.}}
\author{Shayan Oveis Gharan\thanks{\href{mailto:shayan@cs.washington.edu}{shayan@cs.washington.edu}. Research supported by Air Force Office of Scientific Research grant FA9550-20-1-0212, NSF grant  CCF-1907845, and a Sloan fellowship.}} 
\affil{University of Washington}

\begin{document}

\maketitle
\begin{abstract}
	A matroid $\cM$ on a set $E$ of elements has the $\alpha$-partition property, for some $\alpha>0$, 
	if it is possible to (randomly) construct a partition matroid $\PM$ on 
	(a subset of) elements of 
	$\cM$ such that every independent set of $\PM$ is independent in $\cM$ and for any weight function $w:E\to\R_{\geq 0}$, the expected value of the optimum of the matroid secretary problem on $\PM$ is at least an $\alpha$-fraction of the optimum on $\cM$. We show that the complete binary matroid, $\BM_d$ on $\bF_2^d$ does not satisfy the $\alpha$-partition property for any constant $\alpha>0$ (independent of $d$).

	Furthermore, we refute a recent conjecture of \cite{BSY21} by showing the same matroid is $2^d/d$-colorable but cannot be reduced to an $\alpha 2^d/d$-colorable partition matroid for any $\alpha$ that is sublinear in $d$.
\end{abstract}
\newpage

\section{Introduction}

Since its formulation by Babaioff, Immorlica and Kleinberg in 2007~\cite{BIK07, BIKK18}, the matroid secretary conjecture has captured the imagination of many researchers ~\cite{DP08, BDG+09, KP09, IW11, CL12, JSZ13, MTW13, DK13,  Lac14, FSZ15}. This beautiful  conjecture states the following:
Suppose that elements of a known matroid $\cM = (E, \mathcal{I})$ with unknown weights $w:E\to\R_{\geq 0}$ arrive one at a time in a uniformly random order. When an element $e$ arrives we learn its weight $w_e$ and must make an irrevocable and immediate decision as to whether to ``take it" or not, subject to the requirement that the set of elements taken must at all times remain an independent set in the matroid. The {\em matroid secretary conjecture} states that for any matroid, there is an (online) algorithm that guarantees that the expected weight of the set of elements taken is at least a constant fraction of the weight of the maximum weight base. 

More formally, we say the competitive ratio of a matroid secretary algorithm\footnote{That is, an algorithm which decides as elements arrive whether to take them or not.} $A$ on a particular matroid $M$ is
$$\inf_\bw\frac{\E{A_\cM(\bw)}}{\opt_\cM(\bw)}$$
where $A_\cM(\bw)$ is the weight of the set of elements selected by the online algorithm $A$, and $\opt_\cM (\bw)= \max_{I \in \Ind} \sum_{i \in I}w_i.$
We drop the subscript $\cM$ when the matroid is clear in the context.
 The expectation in the numerator is over the uniformly random arrival order of the elements and any randomization in the algorithm itself. The conjecture states that for any matroid, there is an algorithm with competitive ratio $O(1)$.

%The matroid secretary conjecture generalizes the classical result for the standard secretary problem~\cite{Dyn63} for which there is a $1/e$ competitive algorithm. Indeed, the secretary problem is a matroid secretary problem where the matroid is the uniform matroid of rank 1. Note that the classical secretary algorithm immediately gives a $1/e$ competitive algorithm for partition matroids\footnote{A  matroid $M = (E, \mathcal{I})$ is a partition matroid if $E$ can be partitioned into disjoint sets $P_1, \dots, P_t$ such that $\mathcal{I} = \{ S \subseteq E  | \, |S \cap P_i| \leq  1 \text{ for } 1 = 1, \dots,t\} $.} , since one can simply run the standard secretary algorithm in each part.

The matroid secretary conjecture is known to be true for a number of classes of matroids, including partition matroids, uniform matroids, graphic matroids and laminar matroids~\cite{BIK07, BIKK18, DP08, BDG+09, KP09, IW11, JSZ13, MTW13}. In its general form, it remains open. At this time, the best known general matroid secretary algorithm has competitive ratio $O (1/\log\log r)$ where $r$ is the rank of the matroid~\cite{Lac14, FSZ15}.

A reasonably natural approach to proving the matroid secretary conjecture is by  a {\em reduction to a partition matroid}. 
 \begin{definition}
A  matroid $\cM' = (E', \mathcal{I}')$ is a reduction of matroid  $\cM = (E, \mathcal{I})$, if $E' \subseteq E$ and $I' \subseteq I$.
 \end{definition}
A  matroid $\cM $ is a partition matroid if its elements can be partitioned into disjoint sets $P_1, \dots, P_d$ such that $S \subseteq E$ is independent iff $|S \cap P_i| \leq  1$ for all $1\leq i\leq d$. 
Specifically, consider the following class of algorithms: 
\begin{enumerate}
\item Wait until some number of elements have been seen without taking anything. We call this set of elements {\em the sample} and use $S$ to denote this set.
\item Based on the elements in $S$ and their weights, (randomly) reduce $\cM$ to a partition matroid $\PM = P_1 \cup P_2 \cup \dots P_d$ on (a subset of) the non-sample $\overline S$. 
%\item Crucially, it is required that any set of elements $e_1, \ldots, e_d$, where $e_i \in P_i$ forms an independent set in the original binary matroid.
\item In each part $P_i$, run a secretary algorithm which chooses at most one element; e.g. choose the first element in $P_i$ whose weight is above a threshold $\tau_i$ (which may be based on $S$). %partition matroid (never taking any elements that are not in $P$).
\end{enumerate}
%An algorithm of this type can of course be randomized. In this paper, we consider algorithms that choose the sample size from a distribution, and choose the mapping from $S$ and the weights of elements in $S$ to a partition matroid on $\overline S$ randomly.   For the purposes of our discussion, there is no need to include randomization in step 3, due, for example, to the existence of a constant competitive matroid secretary algorithm for partition matroids.
Some appealing applications of this approach which are constant competitive are for graphic matroids \cite{KP09},  laminar matroids, and transversal matroids~\cite{DP08, KP09, JSZ13}. The latter two  algorithms rely crucially on first observing a random sample of elements and then constructing the partition matroid.

%constructed the following partition matroid (using no sample): Pick a random permutation $\pi$ of the vertices of the graph and then construct a "part" for each vertex consisting of all edges that go from that vertex to a vertex that follows it in the order $\pi$. It is immediate that if one edge is selected from each part, the resulting undirected graph is acyclic and hence this is a partition matroid.  A simple analysis then shows that running the standard secretary algorithm in each part yields a $1/2e$ competitive algorithm for graphic matroids.

 %where the set of elements are vectors of dimension $d$ over $\bF_2$, and a set of elements is independent if the corresponding vectors are linearly independent over $\bF_2^d$.

Consider the
 complete {\em binary matroid},  $\BM_d$,  which is the linear matroid defined on all vectors in $\bF_2^d$ where a set $S\subseteq \bF_2^d$  is independent if the vectors in $S$ are linearly independent over the field $\bF_2^d$.
Our main result is that for complete binary matroids, no algorithm of the above type, that is, based on a reduction to a partition matroid, can yield a constant competitive ratio for the matroid secretary problem.
\begin{theorem}[Informal]\label{thm:informal} Any matroid secretary algorithm for complete binary matroids $\BM_d$ that is based on a reduction to a partition matroid has competitive ratio $O(d^{-1/4})$.
\end{theorem}

We say a matroid $\cM = (E, {\cal I})$ has the $\alpha$-partition property if it can be (randomly) reduced to a partition matroid $\PM$ such that
 $$\E{\opt_{\PM}(w)} \geq  \frac1{\alpha} \opt_{\cM}(w).$$
 In a survey \cite{Din2013}, Dinitz raised as an open problem whether  every matroid $\cM$ satisfies the 
 {\em $\alpha$-partition property} for some universal constant $\alpha>0$. %This property says that there is a partition matroid $M' = (E', \Ind')$ with $E' \subseteq E$ and $\Ind' \subseteq \Ind$ such that for every weight vector $\bw$, $$\opt (M', \bw) \ge \frac{1}{\alpha} \opt (M, \bw).$$
Dinitz observes that it is unlikely that the $\alpha$-partition property holds for all matroids, but notes that there is no matroid known for which it is false. 
As a consequence of our main theorem, the complete binary matroid does not satisfy the $\alpha$-partition property for $\alpha\leq O(d^{1/4})$.
%Our result shows that the binary matroid does not have the $\alpha$-partition property for any constant $\alpha$. 
In fact, our negative result is stronger, since it allows for the partition matroid $\PM$ to be constructed after seeing a sample and the weights of the sample. This shows that although this  approach works  for laminar and transversal matroids, it does  not  generalize to all matroids.

As a byproduct of our technique we also refute a conjecture of B\'erczi, Schwarcz, and Yamaguchi \cite{BSY21}.
The covering number of a matroid $\cM = (E,\cal{I})$ is the minimum number of independent sets from ${\cal I}$ needed to cover the ground set $E$. A matroid is $k$-coverable if its covering number is at most $k$.
%
%{ \color{red}  Dorna: I added  "on the same ground set $E$" to the theorem below }

\begin{conjecture}[\cite{BSY21}]\label{conj:cov}
	Every $k$-coverable matroid $\cM = (E, \cal {I})$ can be reduced to a $2k$-coverable partition matroid on the same ground set $E$. 
\end{conjecture}
Below we prove that $\BM_d \setminus \{0\}$  refutes the above conjecture for $d\geq 17$. Note that we need to remove $0$ from the complete binary matroid, since the covering number is not defined for matroids that have loops. 
\begin{restatable}{theorem}{conjref}
\label{thm:conjref}
	For any $d\geq 17$ there exists a matroid $\cM$ of rank $d$ that is $k$-coverable for some $k\geq d$, but it cannot be reduced to a $2k$-coverable partition matroid with the same number of elements. In particular, such an $\cM$ can only be reduced to $\Omega(kd)$ coverable partition matroids.
\end{restatable}
  
\paragraph{Independent Work.} In recent independent work, Bahrani, Beyhaghi,  Singla, and Weinberg \cite{BBSW21} also studied barriers for simple algorithms for the matroid secretary problem. 
We refer the interested reader to \cite{BBSW21} for the details of their contributions.

\section{Main Technical Theorem}
For an integer $k\geq 1$, we write $[k]:=\{1,\dots,k\}$. The following is our main technical theorem:
\begin{theorem}[Main Technical]\label{thm:bad_partitions}
For  any reduction of the complete binary matroid $\BM_d= ( \bF^d_2,\Ind)$ to a partition matroid  $\PM = P_1 \cup \dots \cup P_d$, there is a subset $T \subseteq [d]$ such that  $|T| \geq  d -  8 \sqrt {d}$  and $|\cup_{i \in T} P_i| \leq  \frac{2^d}{\sqrt[4]{d}}$. 
\end{theorem}
Note that,  throughout the paper, for any partition matroid specified by $P_1, \dots, P_d$, we allow sets $P_i$ to be empty. Therefore the partition matroid can effectively have less than $d$ parts and  the reduction does not have to be rank-preserving. 

As a consequence of the above theorem,  there are $O(\sqrt{d})$ parts in $[d]\setminus T $ that contain the vast majority of the elements of $\BM_d$. For appropriately chosen weight vectors, this is bad, since  only $O(\sqrt{d})$ elements can be taken from $\cup_{i \not\in T} P_i$.

We use the  following simple fact.
\begin{fact}
\label{fact:sums} Let $\PM$  be a partition matroid that is reduction of $\BM_d$ with parts $P_1, \ldots, P_d$. Then if two elements $x$ and $y$ are in different parts (say $P_i$ and $P_j$), then their sum $x+y$ is in $P_i$, $P_j$ or $ \bF^d_2 \setminus \cup_i P_i$. %Similarly, if three elements $x,y,v$ are in different parts (say $P_i$, $P_j$ and $P_k$), then their sum $x+y+v$ is either in $P_i$, $P_j$, $P_k$ or 	$ \bF^d_2 \setminus \cup_i P_i$.
\end{fact}

%In the following two lemmas, let $\PM = P_1 \cup \dots \cup P_d$ be a partition matroid that is a reduction of (a subset of) $\mathcal{B}_d$ such that $|\PM|=n=c \cdot 2^d$ for some $0 < c \le 1$. Let $R := \mathbb{F}_2^d \smallsetminus \PM$. (Note that we abuse notation and use $\PM$ to refer to both the partition matroid and the elements of the partition matroid.)

\begin{lemma}
Let $\PM$  be a partition matroid that is reduction of $\BM_d$ with parts $P_1, \ldots, P_d$ and let $R := \mathbb{F}_2^d \smallsetminus \PM$. 
The number of pairs $a \in P_i, b \in P_j$ for $1 \le i < j \le d$ in which $a+b \in R$ is at most $\max_{1 \le i \le d} 2|P_i| \cdot |R|$.
\end{lemma}
\begin{proof}
Create a hypergraph $H$ whose vertices are elements in $\mathbb{F}_2^d$. Now, create a hyperedge $(a,b,a+b)$ for every $a \in P_i, b \in P_j$ for $1 \le i < j \le d$ in which $a+b \in R$. 

Fix any $q \in R$. First, note that there are no two distinct hyperedges $(a,b,q)$, $(a,b',q)$, as this would imply $a+b = a+b'$ and therefore $b=b'$. Therefore, the pairs $(a,b)$ such that $(a,b,q)$ is a hyperedge form a matching. 

Now fix a hyperedge $(a,b,q)$ with $a \in P_i, b \in P_j$. If there is some other edge $(c,d,q)$ such that $c,d \not\in (P_i \cup P_j)$, then $a,b,c,d$ are all in different partitions, which cannot occur as $a+b=c+d$, which is a linear dependence in the partition matroid. Therefore, every edge containing $q$ must contain an element of $P_i \cup P_j$. Therefore the matching contains at most $|P_i|+|P_j| \le 2 \max_{1\leq k\leq d} |P_k|$ edges, from which the claim follows.
\end{proof}

\begin{lemma} \label{lem:main}
Let $\PM$  be a reduction of $\BM_d$ to a partition matroid with parts $P_1,\dots, P_d$ with total size $\sum_{i=1}^d |P_i|=n = c\cdot 2^d$ for some $0 < c \leq 1$. Then, there exists an $1 \leq i \leq d$ such that $|P_i| >  \frac{c}{8}n.$
\end{lemma}
\begin{proof}
By way of contradiction, suppose $\max_{1 \le i \le d} |P_i| \le \frac{cn}{8}$. Now, construct a graph whose vertices are the elements in $P$. First, create an edge $(a,b)$ for all $a \in P_i, b \in P_j, i \not= j$ for which $a+b \not\in R$. 

For each such edge $(a,b)$, by \cref{fact:sums} either $a+b \in P_i$ or $a+b \in P_j$. Direct the edge towards $a$ if the former occurs, otherwise direct it towards $b$. Note that the in-degree of each element in a partition $P_i$ is at most $|P_i| \le \frac{cn}{8}$ (since if $a+b=a+d=a' \in P_i$ then $b=d$). Therefore, there are at most $\frac{cn^2}{8}$ such edges.

However, by the previous lemma, there are at least (using that $\max_i |P_i| \le \frac{cn}{8}$):
\begin{eqnarray*}{n \choose 2} - \sum_{i=1}^d {|P_i|\choose 2} - 2\frac{cn}{8} |R| &\ge&  {n \choose 2} - \sum_{i=1}^d {|P_i|\choose 2} - \frac{n^2}{4} \\
&=& \frac{n^2}{2} - \sum_{i=1}^d \frac{|P_i|^2}{2} - \frac{n^2}{4} > \frac{n^2}{8}
\end{eqnarray*}
such edges, where in the first inequality we used that $|R| \le 2^d = \frac{1}{c}n$, 
in the equality we used $\sum_{i=1}^d |P_i|=n$ and in the last inequality we used that
$\sum_{i=1}^d |P_i|^2$ is maximized when $|P_i| = \frac{cn}{8}$ on $8/c$ parts, and  $c \le 1$, this is a contradiction with the above, which gives the lemma.
\end{proof}

Now, we finish the proof of \cref{thm:bad_partitions}.

\begin{proof}[Proof of \cref{thm:bad_partitions}]
We start from $T = [d]$ and inductively repeat the following: if $|\cup_{i\in T} P_i| \geq \frac{2^d}{\sqrt[4]{d}}$,  remove  $ j = \arg \max_{ i \in T} |P_i|$ from $T$.  Using   \cref{lem:main}, if  $|\cup_{i\in T} P_i| \geq \frac{2^d}{\sqrt[4]{d}}$, the size of the partition that we remove is at least
$$ \max_{ i \in T} |P_i| \geq \frac{1}{8\sqrt[4]{d}} \cdot \frac{2^d}{\sqrt[4]{d}} = \frac{2^d}{8\sqrt{d}}.$$
Therefore,  after at most $8\sqrt{d}$ steps, we get $|\cup_{i\in T} P_i| \leq \frac{2^d}{\sqrt[4]{d}}$. This finishes the proof. 
\end{proof}

\section{Main Theorems}

\subsection{Matroid $\alpha$-Partition Property (Proof of \cref{thm:conjref})}

For a matroid $\cM$, Edmonds defined:
\begin{equation}\label{eq:beta} \beta(M):= \max_{\emptyset \subset F\subseteq E} \frac{|F|}{\rk_M(F)}.	
\end{equation}
Note that the maximum in the RHS is attained at {\em flats} of $\cM$, namely sets $F$ that are the same as their closure.
\begin{theorem}[Edmonds \cite{Edm65}]\label{thm:edmonds}
	For any matroid $\cM$  on elements $E$ and with no loops, the covering number of $\cM$, namely the minimum number of independent sets whose union is $E$ is equal to $\lceil \beta(M)\rceil$.
\end{theorem}

Using this, we show that \cref{thm:conjref} is a corollary of \cref{lem:main}. 

\conjref*
\begin{proof}
It follows from \cref{thm:edmonds} that the  binary matroid $\BM_d\setminus \{0\}$ on $\bF_2^d$ satisfies $\beta(\BM_d\setminus \{0\})=(2^d-1)/d$. This is because the flats of $\BM_d \setminus \{0\}$ correspond to (linear) subspaces. A linear subspace of dimension $k$ has exactly $2^k -1$ many vectors. So, the maximum of \eqref{eq:beta} is attained at $F=\bF_2^d \setminus \{0\}$ which has rank $d$.

Now, suppose $\BM_d\setminus \{0\}$ is reduced to a partition matroid ${\cal P}$ with parts $P_1, \dots, P_d$ such that  $\cup_{i=1}^d P_i = \mathbb{F}_2^d \setminus \{0\}$. Observe that $\beta({\cal P}) = \max_{1\leq i\leq d} |P_i|$. To refute \cref{conj:cov} and prove \cref{thm:conjref}, it is enough to show that $\max_{1\leq i\leq d} |P_i| > \Omega(2^d-1)$. However, by \cref{lem:main} (setting $c=1-1/2^d$ to account for deleting the 0 element), this quantity is at least $\frac{2^d-1}{8}$, which gives the theorem.
\end{proof}

\subsection{Matroid Secretary $\alpha$-Partition Property (Proof of \cref{thm:informal})}

\begin{definition}
Let $\PM(S,\bw|_S)$ be any function that maps  a sample $S\subset \bF_2^d$ and weights $\bw|_S$ of elements in the sample to a partition matroid that is a reduction of $\BM_d$, where the elements of $\PM(S,\bw|_S)$ are a subset of $\overline{S} = \bF_2^d \setminus S$. Let ${\bf P}$ be the collection of all such mappings.

\begin{definition}[Randomized Partition Reduction Algorithm]A (randomized) {\em partition reduction algorithm} $\Alg$ for a matroid $\cM$ with $n$ elements consists of two parts:
\begin{itemize}
\item $\Alg$ (randomly) chooses a sample size $0\leq |S|\leq n$  before any elements have been seen; we denote this choice by $s_\Alg$.
\item $\Alg$ (randomly) chooses a mapping $\PM_{\Alg} \in {\bf P}$ and uses it to build $\PM(S,\bw|_S)$ after seeing the sample  $S$. \end{itemize}
\end{definition}
\end{definition}
%Our main theorem is the following:
\begin{theorem}[Main]\label{thm:main}
For any randomized partition reduction algorithm  $\Alg$ for $\BM_d$, with $d \geq 2^{12}$ there is a weight function $\bw:  \bF^d_2 \rightarrow \mathbb{R}_{\geq 0}$ such that

$$\mathbb{E}_{s_\Alg}  \mathbb{E}_{S:|S| = s_\Alg} \mathbb{E}_{\PM_\Alg}[\opt_{\PM_\Alg}(\bw|_{\overline S})] \leq 4 d^{-\frac{1}{4}} \opt_{\BM_d}(\bw).$$
For readability in the above, we have suppressed the fact that the partition matroid  $\PM_\Alg$  (whose elements are a subset of $\overline{S}$) depends on both $S$ and $\bw|_S$. Note that $S$ is drawn from the uniform distribution over subsets of $\bF_2^d$ of size $s_\Alg$.

\end{theorem}

%The crucial theorem that we will be using is the following.

%Now, we are ready to apply   \cref{thm:bad_partitions} to prove \cref{thm:main}. 

\begin{proof}%[Proof of \cref{thm:main}]
Suppose that the weights of the elements in $\BM_d$ are selected by setting 
\begin{equation}
\label{weights}	w_i = \bone_{i \in X}
\end{equation}
 where $X=\{x_1, x_2, \ldots, x_d\}$ is a uniformly random sample of $d$ elements from $\bF_2^d$, selected with replacement. Then the optimal independent set has expected weight $ \mathbb{E}_\bw \opt_{\BM _d}(\bw)$ equal to 
\begin{equation}
\label{rankX}\mathbb{E}_X(\rk(X )) = \sum_{i=1}^d \mathbf{1}_{x_i \notin \spn \{x_1, \dots, x_{i-1}\} }= \sum_{i=1}^d \frac{2^d - 2^{i-1}}{2^d - (i-1)}  \geq \sum_{i=1}^d \frac{2^{d-1}}{2^d} = \frac{d}{2}.
\end{equation}
%This implies that at least 1/3 of the sets of size $d$ have rank at least $d/4$.  
%
%{ \color{red}  Dorna:  this part was a bit  confusing to me so I changed it a bit. The previous version is commented. Please feel free to change it back if you think the previous one was better}

%Now let $\Alg = (s_\Alg, \PM_\Alg)$ be an algorithm that given a sample of size $s_\Alg$  along with $\bw|_S$ , constructs a partition matroid $\PM$ on the elements of $\overline S= \bF_2^d \setminus S$. 
Now let $\Alg$ be an arbitrary algorithm that  chooses the sample size $s_\Alg$ and the mapping  $\PM \in {\bf P}$ deterministically. 
We claim that it  suffices to show that  for the weight vector given in \cref{weights} when $X$ is chosen uniformly at random, and for $S$ a uniformly random sample of elements of any fixed size:
\begin{equation}\label{bad}\mathbb{E}_\bw \mathbb{E}_S \opt_{\PM}(\bw|_{\overline{S}})  \le 2 d^{3/4},\end{equation}
where $\PM = \PM(S, \bw|_S)$ is any partition matroid on a subset of $\overline S$ constructed after seeing the elements in $S$ and their weights. (Note that $\opt_{\PM}$ is a upper bound on the performance of $\Alg$.)

%{ \color{red}  Dorna: also modified the above sentence slightly to make it compatible with the  above change. }

%To see why, observe that by taking the expected value over the randomization in $\Alg$ and then interchanging the order of the expectations, we get
To see why, observe that  in the randomized case, by taking the expected value over the randomization in $\Alg$ and then interchanging the order of the expectations, we get
$$ \mathbb{E}_{\bw}\mathbb{E}_{\Alg}  \mathbb{E}_S \opt_{\PM_\Alg}(\bw|_{\overline{S}})  = \mathbb{E}_{\Alg}  \mathbb{E}_{\bw} \mathbb{E}_S \opt_{\PM_\Alg}(\bw|_{\overline{S}})   \le 2 d^{3/4}.$$
%By Markov's inequality, this implies that for at least 3/4 of the sets $X$ of size $d$ (that determine $\bw$ as per \eqref{weights}), we have 
%$$\mathbb{E}_{\Alg}  \mathbb{E}_S \opt_{\PM_\Alg}(\bw|_{\overline{S}})  \le 8 d^{3/4}$$
Therefore, for $\bw$ chosen at random according to (\ref{weights}) and using (\ref{rankX}),
$$\frac{\mathbb{E}_{\bw}\mathbb{E}_{\Alg}  \mathbb{E}_S \opt_{\PM_\Alg}(\bw|_{\overline{S}}) }{ \mathbb{E}_\bw \opt_{\BM _d}(\bw)} \le 4d^{-1/4}.$$ Applying the mediant inequality, we conclude that there is a set $B$ of size $d$ such that $\frac{\mathbb{E}_{\Alg}  \mathbb{E}_S \opt_{\PM_\Alg}(B|_{\overline{S}})}{\rk(B)}$ is at most  $4 d^{-1/4}$,
completing the proof of the theorem.

It remains to prove \eqref{bad}. Observing that $\bw$ (resp. $\bw|_S$, $\bw|_{\overline{S}}$) is fully determined by $X$ (respectively $X \cap S$, $X \cap \overline{S}$) and letting $X_1 := X \cap S$, $X_2 := X \cap \overline{S}$ we write  
\begin{equation}
\label{eq:rearrange_opt_small}
	\mathbb{E}_\bw \mathbb{E}_S \opt_{\PM}(\bw|_{\overline{S}})    = \mathbb{E}_S \mathbb{E}_{X_1} \mathbb{E}_{X_2}\opt_{\PM}(X_2).
\end{equation}

For any choice of $S$ and $X_1$, the partition matroid $\PM=\PM(S, X_1)$ on a subset of $\overline S$ consists of parts $ P_1 \cup \dots \cup P_d$ (some of these parts could be empty). By \cref{thm:bad_partitions}, there exists a set $T \subseteq [d]$ of size at least $d- 8\sqrt{d}$ such that $|\cup_{i \in T} P_i| \leq \frac{2^d}{\sqrt[4]{d}}$. 
Therefore, $$\opt_{\PM}(X_2)    \le |X_2 \cap \cup_{i \in T} P_i| +8 \sqrt{d}.$$
So, for a fixed $S$, we have
\begin{align*}
\mathbb{E}_{X_1}\mathbb{E}_{X_2} \opt_{\PM}(X_2)
&\le \mathbb{E}_{X_1}\mathbb{E}_{X_2} \left(|X_2 \cap \cup_{i \in T} P_i| + 8 \sqrt{d} \right)
\\
& =\mathbb{E}_{X_1}\left(  (d-|X_1|)\frac{ |\cup_{i \in T} P_i| }{|\overline{S}|} + 8 \sqrt{d}\right). 	\\
& \le  \mathbb{E}_{X_1}\left(  (d-|X_1|)\frac{ 2^d/\sqrt[4]{d} }{|\overline{S}|} + 8 \sqrt{d}\right). 	\\
& =(d-\mathbb{E}(|X_1|))\frac{2^d }{\sqrt[4]{d} |\overline{S}|} + 8 \sqrt{d}. 	
\end{align*}
Finally, we observe that 
$$ \frac{(d-\mathbb{E}(|X_1|))}{|\overline{S}|} \cdot \frac{2^d }{\sqrt[4]{d}} = \frac{d-\frac{|S|\cdot d}{2^d}}{(1-\frac{|S|}{2^d})2^d}\cdot \frac{2^d}{d^{1/4}} = d^{3/4}.$$
Thus, we get  $$\mathbb{E}_{X_1}\mathbb{E}_{X_2} \opt_{\PM}(X_2)  \le d^{3/4}  + 8\sqrt{d} \leq 2d^{3/4},$$ 
where in the last inequality we used our assumption that $d \geq 2^{12}$. Combining \cref{eq:rearrange_opt_small} with this, \cref{bad} follows. 
\end{proof}

\section{Conclusion}

We note that for our bad example, the trivial algorithm for matroid secretary succeeds: one simply needs to take every improving element when it arrives. An interesting open problem which appears approachable is whether there exists an example which simultaneously fails for any partition reduction as well as for the trivial algorithm. One can also more generally try to refute strengthened versions of the partition algorithm.
%An interesting open direction is to formalize and refute strengthened versions of the partition algorithm; a specific question which appears approachable is whether there exists an instance which simultaneously fails for any partition reduction as well as for the trivial algorithm.
%One 
\printbibliography

\end{document}